\documentclass[12pt,english]{article}
\usepackage[T1]{fontenc}
\usepackage[latin9]{inputenc}
\usepackage[a4paper]{geometry}
\usepackage{color}
\usepackage{babel}
\usepackage{amsmath}
\usepackage{amsthm}
\usepackage{amssymb}
\usepackage{graphicx}
\usepackage{setspace}
\usepackage[authoryear]{natbib}
\usepackage[unicode=true,pdfusetitle,
 bookmarks=true,bookmarksnumbered=false,bookmarksopen=false,
 breaklinks=false,pdfborder={0 0 1},backref=false,colorlinks=true]
 {hyperref}
\hypersetup{
 linkcolor=blue, citecolor=blue}

\makeatletter
\theoremstyle{plain}
\newtheorem{thm}{\protect\theoremname}
\theoremstyle{plain}
\newtheorem{lem}{\protect\lemmaname}
\theoremstyle{remark}
\newtheorem{rem}{\protect\remarkname}
\theoremstyle{plain}
\newtheorem{prop}{\protect\propositionname}
\theoremstyle{plain}

\theoremstyle{plain}
\newtheorem{claim}{\protect\claimname}

\usepackage{tikz}
\usepackage{pgfplots}
\usepackage{mathdots}

\def \co{ {\rm co\,} }

\def \cav {\textrm{cav\,}}

\usepackage[textwidth=30mm]{todonotes}

\@ifundefined{showcaptionsetup}{}{%
 \PassOptionsToPackage{caption=false}{subfig}}
\usepackage{subfig}
\makeatother

\providecommand{\lemmaname}{Lemma}
\providecommand{\propositionname}{Proposition}
\providecommand{\remarkname}{Remark}
\providecommand{\theoremname}{Theorem}
\providecommand{\corollaryname}{Corollary}
\providecommand{\claimname}{Claim}

\begin{document}
\title{Cross-verification and Persuasive Cheap Talk\thanks{We are grateful for the helpful comments of Navin Kartik, Elliot Lipnowski and Joel Sobel. Ludovic Renou gratefully acknowledges the support of the Agence Nationale pour la Recherche under grant ANR CIGNE (ANR-15-CE38-0007-01) and through the ORA Project ``Ambiguity in Dynamic Environments'' (ANR-18-ORAR-0005). Atakan's work on this project was supported by a grant from the European Research Council (ERC 681460, InformativePrices).}}
\author{Alp Atakan\thanks{Koc University and QMUL. } \,\& Mehmet Ekmekci\thanks{ Boston College}
\,\& Ludovic Renou \thanks{QMUL and CEPR}}

\maketitle
\begin{abstract}

We study a cheap-talk game where two experts first choose what information to acquire and then offer advice to a decision-maker whose actions affect the welfare of all. The experts cannot commit to reporting strategies. Yet, we show that the decision-maker's ability to cross-verify the experts' advice acts as a commitment device for the experts. We prove the existence of an equilibrium, where an expert's equilibrium payoff is equal to what he would obtain if he could commit to truthfully revealing his information.
\bigskip{}


\end{abstract}

\newpage

\section{Introduction}
Decision-makers routinely solicit advice from experts who have a vested interest in the decision at hand. Consulting multiple experts may allow a decision-maker to check the veracity of the advice that he receives by comparing one expert's recommendation with another's ("cross-verification"). In this paper, we study how cross-verification affects communication.

\medskip

Cross-verification's effectiveness depends on the experts' information. If experts have perfectly correlated information, then inconsistent recommendations from experts
definitively indicate untruthful, self-serving advice. Alternatively, if experts have uncorrelated information, then cross-verification cannot detect misleading advice. Thus, if the experts strategically acquire information,  their choices will affect the scope for cross-verification. This paper sheds light on this interplay by analyzing a cheap-talk game, where experts independently acquire information before
providing advice to a decision-maker. More precisely, we study the following game: Two experts with
identical preferences, first choose statistical experiments that provide
information about an unknown state of the world, privately observe their experiments' outcomes, and then offer private
reports to the decision-maker. The decision-maker collects
all the reports and chooses an action. \medskip

As a benchmark, suppose that an expert could commit to revealing his experiment's outcome truthfully. Following  \citet{kamenica2011bayesian}, we call the experiment that this expert would optimally select the expert-optimal experiment. In our model, however, the experts
\emph{cannot} commit. Yet, we show that there exists an equilibrium, where both experts
choose the expert-optimal experiment and truthfully report
the outcomes of their experiments. In equilibrium, the experts optimally select perfectly correlated experiments, which enable cross-verification to be most effective. In turn,  cross-verification facilitates truthful communication and allows
the experts to receive their best possible payoff. In other words, cross-verification acts as a commitment device.\medskip

The existence of such an equilibrium relies on three essential properties.
First, we assume that the experts are free to choose arbitrarily correlated statistical
experiments (see \citet{green1978two} and \citet{gentzkow2016competition,gentzkow2017bayesian}). In fact in the equilibrium that we construct, they choose to correlate
their experiments' outcomes perfectly and thus allow the decision-maker
to cross-verify their reports perfectly.
Second, suppose one expert deviates from reporting the experiment's outcome 
truthfully, while the other is truthful. In this case, the decision-maker detects
a deviation as the two reports are inconsistent.
However, the decision-maker cannot deduce the deviator's identity.
Third, we show that a \emph{uniform} punishment always exists. There is an action that punishes the experts for deviating from
truthful reporting, irrespective of the experts'
private information. \medskip

The existence of the aforementioned uniform punishment is key to our
equilibrium construction since the decision-maker does not know the
deviator's identity and, therefore, cannot condition the punishment on
the deviator's information. Proving the existence of a uniform punishment is the main technical contribution
of the paper. We stress that the uniform
punishment is relative to the expert-optimal experiment. Arbitrary
experiments do not necessarily admit uniform punishments, and therefore,
cross-verification does not necessarily elicit honest advice when the experts choose arbitrary experiments.\medskip{}

Our main result, described above, also generalizes to situations
where the experts have non-identical preferences, provided that a
uniform punishment continues to exist. In particular, we show that
there is a uniform punishment when the preferences of the second expert
are a convex combination of the preferences of the first expert and
the decision-maker. For example, this is the case in the 
quadratic utility example of \citet{crawford1982strategic} when the two experts
have like-biases. \medskip

Finally, we also study cross-verification from the decision-maker's perspective. We show that there is an equilibrium where the
decision-maker benefits from cross-verification if the expert-optimal
experiment is informative at some prior belief.  The intuition is as follows:  The decision-maker benefits from any additional information, and even the expert-optimal experiment provides
valuable information in many circumstances. If the expert-optimal experiment does not provide
useful information to the decision-maker, we appropriately modify
the expert-optimal experiment. The modified experiment offers
valuable information for the decision-maker, and the experts can truthfully
communicate this information in equilibrium. We also establish this result's converse: the decision-maker's unique equilibrium payoff is equal
to his payoff at his prior belief if the expert-optimal experiment is uninformative at every prior belief. In other
words, the decision-maker only benefits from cross-verification in
situations where the experts also benefit.\medskip{}

\textbf{\textit{Related literature.}} This paper is related
to the literature on cheap talk pioneered by \citet{crawford1982strategic}
and several papers in this literature study communication with
multiple experts. In particular, \citet{krishna2001asymmetric} focus
on a model where the experts are perfectly informed and show that
there is an equilibrium where the experts truthfully reveal the state
if the experts send messages simultaneously. In contrast, \citet{krishna2001model}
prove that such an equilibrium does not exist if the experts send
messages sequentially.  \citet{battaglini2002multiple} shows that the decision-maker can learn a multidimensional state by consulting experts about different dimensions. \citet{ambrus2014almost} find that there are equilibrium outcomes of multi-sender cheap-talk games that are arbitrarily close to full revelation, if the senders imperfectly observe the state and if the state space is large enough.\footnote{Also, see \citet{wolinsky2002eliciting},
and \citet{gilligan-krehbiel} for related work on multi-sender cheap-talk games.} Our work differs from these articles in several respects: Foremost, our main result shows that the experts obtain their commitment
payoff. In contrast, the cheap-talk literature is predominantly interested in full information revelation. In other words, our emphasis is on the experts' perspective while the cheap-talk literature focuses on the decision-maker's perspective. Second, we assume that the experts choose what kind of
information to acquire, while the previous papers typically assume that
the experts perfectly know the state. This is an important distinction
since the experts' information affects the scope for cross-verification.
Third, the literature on cheap talk focuses on agents with single-peaked preferences and frequently assumes that all agents have quadratic utility.
In contrast, we put no restrictions on the utility functions.\footnote{With quadratic utility and like-biased experts, the expert-optimal
information structure coincides with the decision-maker's and entails choosing the perfectly informative experiment.
Therefore, as in  \citet{krishna2001asymmetric}, our result also implies that full information revelation
is an equilibrium in this particular case. However, with other utility specifications,
the expert-optimal and decision-maker optimal  information structures need not coincide.} \medskip

The
survey by \citet{sobel2013giving} also discusses how cross-verification
ensures truth-telling in the context of multi-sender
cheap-talk games. The argument provided in this survey relies on the
existence of an arbitrarily harsh exogenously-given punishment for
deviations from truthful reporting. Instead, we show that a uniform
punishment relative to the optimal experiment always exists.\medskip

Our paper is also related to the following works that focus on single-expert cheap-talk games: \citet{lyu2020information} characterizes the equilibrium set in a model where the expert acquires information before providing advice. \citet{lipnowski2020equivalence} shows that an expert can obtain his commitment payoff if the expert's value function is continuous.\footnote{The value function describes the expert's highest expected payoff at a given belief conditional on the decision-maker choosing a best-reply to that belief. Continuity of the value function is a strong assumption. E.g.,  with two states and two actions, it requires the expert to be indifferent between the two actions whenever the decision-maker is.} Instead, we focus on a model with multiple experts and show that the experts receive their commitment payoff, without making any assumptions on their payoff functions. \medskip

Finally, this paper is closely related to the literature on Bayesian persuasion
(\citet{kamenica2011bayesian}). A number of articles that include \citet{au2020competitive},
\citet{gentzkow2016competition,gentzkow2017bayesian}, \citet{koessler2018interactive}
and \citet{li2018bayesian,li2018sequential} study persuasion with
multiple experts. In all of these papers, the experts can commit to
revealing their information truthfully. In contrast, we assume that the experts' recommendations are cheap-talk, i.e., we require sequential
rationality at every stage of the game. Our result
shows that the experts can achieve their commitment payoff even though
they cannot commit to revealing their information.
For a recent survey of the literature on Bayesian persuasion, we refer to \citet{kamenica2019bayesian}.

\section{The Model}

We study a cheap-talk game between two experts, labelled 1 and 2,
and a decision-maker. The experts provide information to the decision-maker
about a payoff-relevant state $\omega\in\Omega$, who then chooses
an action $a\in A$. The sets $A$ and $\Omega$ are finite. The
experts have identical preferences. An expert's payoff is $u(a,\omega)$
when the decision-maker chooses action $a$ and the state is $\omega$.
(We relax the identical preferences assumption in the next section.)
The decision-maker's payoff is $v(a,\omega)$. Initially, neither
the experts nor the decision-maker knows the state. The common prior
probability that the state is $\omega$ is $\pi^{\circ}(\omega)$.
\medskip

We first provide an informal description of the cheap-talk game. The
game has three stages. In the first stage, the two experts simultaneously
choose a statistical experiment. The selected experiments are publicly
observed. In the second stage, each expert privately observes his experiment's outcome and then sends a message to the decision-maker.
In the third stage, the decision-maker observes the experts' messages
and chooses an action.\medskip{}

We now provide a formal description. To model the choice of statistical
experiments, we follow \citet{gentzkow2016competition,gentzkow2017bayesian}.
These authors define a statistical experiment $\sigma$ as a partition
of $\Omega\times[0,1]$ into finitely many (Lebesgue) measurable subsets. A signal $s$ is an element of the partition
$\sigma$, i.e., a measurable subset of $\Omega\times[0,1]$. The
probability of signal $s\in\sigma$ conditional on $\omega$ is the
(Lebesgue) measure of the set $\{x\in[0,1]:(\omega,x)\in s\}$. Throughout,
we omit the dependence on the experiment $\sigma$, and write $\lambda_{s}$
for the probability of the signal $s$ and $\pi_{s}$ for the posterior
probability. We denote the set of experiments that the experts can
choose from by $\Sigma$.\medskip{}

In the first stage, expert $i$ thus chooses an experiment $\sigma_{i}\in\Sigma$.
The chosen experiments $(\sigma_{1},\sigma_{2})$ are publicly observed.
In the second stage, expert $i$ privately observes the realization
$s_{i}\in\sigma_{i}$ and sends a private message $m_{i}\in M_{i}$
to the decision-maker. We assume that the sets of messages are rich
enough to communicate any signal realizations. Finally, the decision-maker
observes the messages $(m_{1},m_{2})$ (but not the realized signals
$(s_{1},s_{2})$) and chooses an action $a$. We denote $\Gamma(\pi^{\circ},u,v)$
the cheap-talk game. Note that different extensive-form games are
consistent with our description. Throughout, we assume that the state
$(\omega,x)\in\Omega\times[0,1]$ is chosen by Nature according to
the probability distribution $\pi^{\circ}\times U[0,1]$ after the
experts have chosen their experiments where $U[0,1]$ denotes the uniform distribution on the unit interval. Thus, we have a proper sub-game
after each choice of statistical experiments $(\sigma_{1},\sigma_{2})$.\medskip{}

A strategy for expert $i$ is a pair $(\sigma_{i},\tau_{i})$, where
$\sigma_{i}\in\Sigma$ and $\tau_{i}(\sigma_{i},\sigma_{j},s_{i})\in\Delta(M_{i})$
for all $(\sigma_{i},\sigma_{j},s_{i})$ with $s_{i}\in\sigma_{i}$.
A strategy for the decision-maker specifies a mixed action $\alpha(\sigma_{i},\sigma_{j},m_{i},m_{j})\in\Delta(A)$
for all $(\sigma_{i},\sigma_{j},m_{i},m_{j})$.\footnote{To ease exposition, we do not explicitly consider randomizations over
the choices of experiments. This does not
affect any of our results.} The solution concept is weak perfect Bayesian equilibrium. We stress
that this requires the beliefs to be consistent with the chosen experiments
$(\sigma_{1},\sigma_{2})$ even if these experiments are off the equilibrium
path.

\medskip{}

Few remarks are worth making. First, as in classical cheap-talk games,
none of the experts can commit to reporting strategies.
Second, if the experiments are $(\sigma_{1},\sigma_{2})$, then the
joint probability of $(s_{1},s_{2})\in\sigma_{1}\times\sigma_{2}$
conditional on $\omega$ is the measure of the set $\{x:(\omega,x)\in s_{1}\cap s_{2}\}$.
Thus, if both experts choose the same experiment $\sigma$, then the
probability of $(s,s')\in\sigma\times\sigma$ is zero, whenever $s\neq s'$.
(To see this, note that if $s\neq s'$, then $s\cap s'=\emptyset$
since $\sigma$ is a partition.) In words, if both experts choose
the same experiment, their realized signals are perfectly correlated.
This property will turn out to be crucial.\footnote{Note, however, that we can allow for the experts to choose identical
and independent experiments without affecting our results. To do so,
it suffices to define an experiment as a finite partition of $\Omega\times[0,1]\times[0,1]$,
with $(\omega,x,y)$ distributed according to $\pi^{\circ}\times U([0,1])\times U([0,1])$.
Intuitively, if the experts condition their random observations on
$x$, they are perfectly correlated, while they are independent if
one expert conditions on $x$ and the other on $y$.} An alternative modeling is to assume that there is a fixed set of
statistical experiments and let the experts observe the realization
of the experiment of their choices. This alternative modeling also
implies that if the two experts choose to observe the same experiment's realization, their observations are identical. Lastly,
it is usual to model statistical experiments as probability kernels
$\sigma^{*}:\Omega\rightarrow\Delta(S)$, where $S$ is the (finite)
set of signals. The latter formulation naturally implies the former:
for each $\omega$, we can partition $[0,1]$ into $|S|$ non-empty
and disjoint intervals such that the length of the $s$-th interval
is $\sigma^{*}(s|\omega)$ when the state is $\omega$. With a slight
abuse of notation, we identify the probability kernel $\sigma^{*}$
with that particular partition of $\Omega\times[0,1]$. \medskip{}

We focus on \textbf{truthful equilibria}, in which
the two experts choose the same experiment in the first stage and
truthfully report the common signal realization in the second stage.\medskip{}

In what follows, we denote by $v(\alpha,\pi)$
 the decision-maker's expected payoff when he chooses the mixed
action $\alpha$ and his belief is $\pi$, and by $BR(\pi):=\{\alpha\in\Delta(A):v(\alpha,\pi)\geq v(\alpha',\pi),\forall\alpha'\in\Delta(A)\}$
 the set of decision-maker's best-replies at $\pi$. Similarly,
we write $u(\alpha,\pi)$ for an expert's expected payoff.

\section{The Main Result}

In this section, we show that the ability of the decision-maker to cross-verify information
serves as a commitment device for the experts. More precisely, we show
that there exists an equilibrium of the cheap-talk game, in
which the experts obtain their \emph{commitment value}.\medskip{}

We define the commitment value as the highest payoff an expert can obtain when he commits to truthfully disclose the realized signal, as in games of Bayesian persuasion. Formally, consider the persuasion game, where an expert first chooses
a statistical experiment $\sigma:\Omega\rightarrow\Delta(S)$, \emph{commits}
to truthfully reveal the realized signal $s$ to the decision-maker, who then
makes a decision. \citet{kamenica2011bayesian} prove
that the best equilibrium payoff for the expert in this game is given
by $\text{cav}\;\overline{u}(\pi)$, where $\text{cav}\;\overline{u}$
is the concavification of $\overline{u}$ and $\overline{u}(\pi):=\max_{\alpha\in BR(\pi)}u(\alpha,\pi)$. (See also \citealp{AumannMaschler95}.) 
 For later reference, we write $(\lambda^{*}_s,\pi_{s}^{*})_{s\in S}$ for an optimal splitting of the prior $\pi^{\circ}$, that is,  $\sum_{s\in S}\lambda^{*}_s\overline{u}(\pi_s^{*})=\text{cav}\,\overline{u}(\pi^{\circ})$ and
$\sum_{s\in S}\lambda^{*}_s\pi_{s}^{*}=\pi^{\circ}$. We write $\Pi^{*}$ for $\{\pi_{s}^{*}:s\in S\}$,
$\mathrm{co\,}\Pi^{*}$ for the convex hull of $\Pi^{*}$, and 
$\Delta^{*}$ for  the set of all probability distributions 
over $\Pi^{*}$. The corresponding optimal experiment is denoted $\sigma^*$.

\medskip

\begin{thm}
There exists a truthful equilibrium of the cheap-talk game, where both experts obtain their commitment value 
$\text{cav\;}\overline{u}(\pi^{\circ})$.
\end{thm}

Before proving Theorem 1, we explain our result's logic with the help of a simple example. There are two states, $\omega_0$ and $\omega_1$, and four actions,  $a_{0},a_{L},a_{R}$ and $a_{1}$. The preferences
are depicted in Figure \ref{ex:mixed strategy punishment}. Throughout the example, probabilities refer to the probability of $\omega_1$. Assume that $\pi^{\circ}=0.45$.

\begin{figure}[h]
\subfloat[{DM's Preferences. This figure depicts $v(a,\pi)$ as a function of
$\pi=\Pr[\omega=\omega_{1}]$. Action $a_{0}$ is optimal for the
DM for $\pi\in[0,0.3]$, action $a_{L}$ is optimal for $\pi\in[0.3,0.4]$,
$a_{R}$ is optimal for $[0.4,0,6]$, and $a_{1}$ is optimal for
$[0.6,1]$.}]{\includegraphics[scale=0.9]{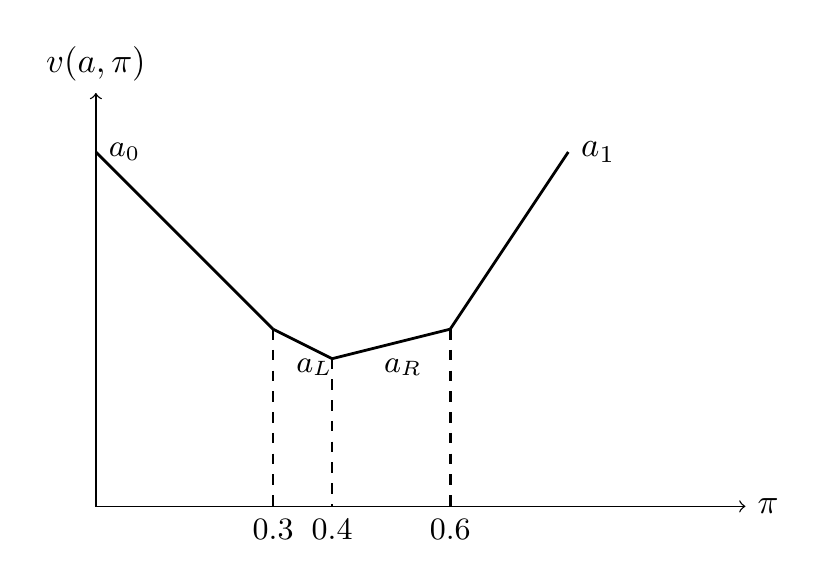}} \qquad{}\subfloat[{Expert's Preferences. The black dashed lines depict $u(a,\pi)$ as a function
of $\pi$. The solid black lines depict $\overline{u}(\pi)$, the
solid blue line depicts $\text{cav }\overline{u}(\pi)$ for $\pi\in[0,3,0.6]$,
and $\text{cav }\overline{u}(\pi)$ coincides with $\overline{u}(\pi)$
for $\pi\protect\notin[0,3,0.6]$. The red dashed line depicts the
payoff to the mixed action $\alpha_{p}\in\Delta(\{a_{L},a_{R}\})$. The mixed action $\alpha_p$ is the uniform punishment; it is a best response for the DM to belief $\pi_{p}=0.4$.}]{\includegraphics[scale=0.9]{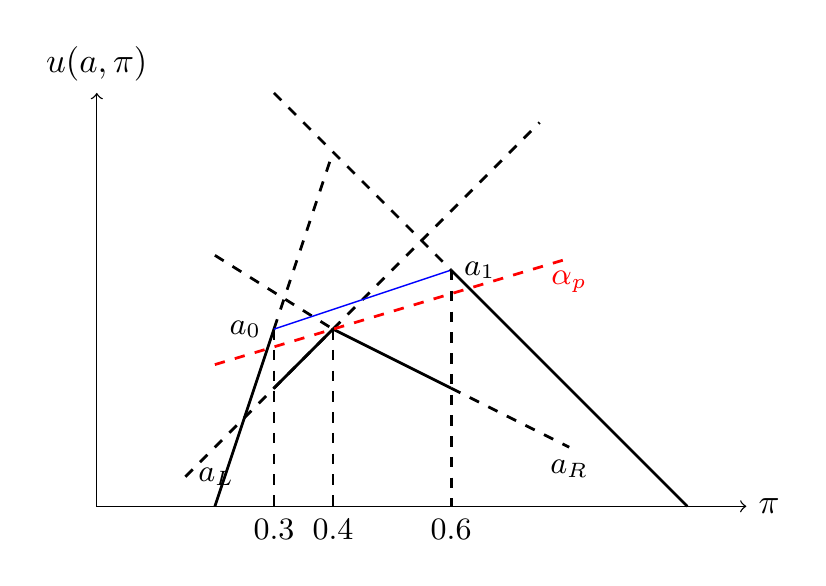}}

\caption{Uniform Punishment and Cross-verification.}
\label{ex:mixed strategy punishment}
\end{figure}
\medskip

We first note that the optimal experiment $\sigma^*$ consists in splitting the prior into the posteriors 
posteriors $\pi_{s_{0}}^{*}=0.3$ and
$\pi_{s_{1}}^{*}=0.6$.\footnote{We have $\lambda^{*}_{s_{0}}=\lambda^{*}_{s_{1}}=1/2$. The experiment is given by: $\sigma^*(s_{0}\vert\omega_{0})=0.64$
and $\sigma^*(s_1|\omega_1)= 0.67$.} We also note that $u(a_0,\pi_{s_0}) < u(a_1,\pi_{s_0})$ and $u(a_1,\pi_{s_1}) < u(a_0,\pi_{s_1})$, that is, the experts have an incentive to mis-report the realized signals. Thus, if there was a single expert, choosing the experiment $\sigma^*$ and truthfully reporting the realized signal would not be an equilibrium. More generally, no equilibrium would give the expert his commitment value. \medskip

Matters are different if the decision-maker chooses to consult  another expert.  To see this, suppose that the two experts choose the experiment $\sigma^*$ and truthfully report the outcome of the experiment.  The decision-maker then holds belief
0.3 (resp., 0.6) and plays action $a_{0}$ (resp., $a_{1}$) after observing
two matching messages equal to $s_{0}$ (resp., $s_{1}$). 
Off the equilibrium path, i.e., when  the decision-maker observes
two contradictory messages, assume that he holds belief $\pi_{p}=0.4$ and
plays action $\alpha_{p}\in\Delta(\{a_{L},a_{R}\})=BR(0.4)$. \medskip

The key observation to make is that the mixed strategy $\alpha_{p}\in BR(0.4)$
is a \emph{uniform punishment}, that is, $u(\alpha_p, \pi_{s_0}) < u(a_0,\pi_{s_0})$
 and  $u(\alpha_p, \pi_{s_1}) < u(a_1,\pi_{s_1})$.  (See Figure \ref{ex:mixed strategy punishment}.) In words, regardless of the realized signal, an expert is punished for deviating from truth-telling. All the decision-maker needs to know is that a deviation has occurred, and  the presence of the second expert indeed guarantees that deviations are detected. The experts thus benefit from the decision-maker cross-verifying their information. (Naturally, there are other equilibria, where the decision-maker benefits from cross-verification. See the next section.)

\medskip 

We conclude with two additional remarks. First,  if the experts choose the perfectly informative experiment,
 truthful reporting does not constitute an equilibrium. This is because
the actions that are best for the decision-maker at beliefs $\pi_{s_{0}}=0$ and
$\pi_{s_{1}}=1$ are the worst for the experts at those beliefs. Second, for any two experiments $\sigma_1$ and $\sigma_2$, 
there is an equilibrium,  where experts 1 and
2 choose experiments $\sigma_{1}$ and $\sigma_{2}$, respectively, and a babbling
equilibrium of the ensuing sub-game is played. 
\medskip

We now turn to the proof of Theorem 1. The proof rests on three essential properties. First, if the two experts choose the same experiment, their signals' realizations are \emph{perfectly correlated}.  This is because they observe the same outcome. Second,  if the two experts choose the same experiment, the decision-maker \emph{detects} any deviation from truth-telling. This is because the decision-maker receives contradicting messages after any deviation.  However, he cannot identify the deviator and, thus, cannot infer the true signal's realization. Therefore, to deter deviations, the decision-maker must be able to  punish the two experts simultaneously. The third property is the existence of such a \emph{uniform punishment} whenever the experiment is expert optimal. The following lemma states this property.

\begin{lem}[Uniform punishment]
Let $(\lambda_s^*,\pi_s^*)_{s \in S}$ be an optimal splitting.  There exist $\pi_{p}\in\text{co }\Pi^{*}$
and $\alpha_{p}\in BR(\pi_{p})$ such that $u(\alpha_{p},\pi_{s}^{*})\leq\overline{u}(\pi_{s}^{*})$
for all $\pi_{s}^{*}\in\Pi^{*}$. 
\end{lem}

Lemma 1 is our main technical contribution. We postpone its proof to the end of this section and  now show how to construct an equilibrium of the cheap-talk game with a payoff of $\text{cav\,}\overline{u}(\pi^{\circ})$ to the experts.

\begin{proof}[Proof of Theorem 1.]
Let $(\lambda_s^*,\pi_s^*)_{s \in S}$ be an optimal splitting inducing the payoff $\text{cav\,}\overline{u}(\pi^{\circ})$. Let $\sigma^*$ be the optimal experiment associated with that splitting. Recall that 
$\Pi^*:=\{\pi_s^*: s\in S\}$. From Lemma 1, there exist $\pi_{p}\in \mathrm{co\,} \Pi^*$ and $\alpha_{p}\in BR(\pi_{p})$
such that for all $\pi_{s}^{*}\in\Pi^{*}$, $u(\alpha_{p},\pi_{s}^{*})-\overline{u}(\pi_{s}^{*})\leq0$.\medskip 

We construct a truthful equilibrium as follows. The experts choose the optimal experiment
$\sigma^{*}$ and truthfully report the realized signal. Following the
choice of $\sigma^{*}$, the decision-maker chooses $\alpha\in BR(\pi_{s})$, with
$u(\alpha,\pi_{s})=\overline{u}(\pi_{s})$, when he observes two
identical messages equal to $s$. Alternatively, if the decision-maker receives two
conflicting messages, he chooses $\alpha_{p}$ (sustained by the
belief $\pi_{p}$). Finally, following the choice of any other statistical
experiment, an equilibrium of the continuation game, which exists by finiteness,
is played. It is routine to check that we indeed have an equilibrium.
\end{proof}

We now offer a series of remarks.

\begin{rem}
We have assumed that the two experts share the same preferences. If the preferences of one expert, say the second expert, 
are a convex combination of the preferences of the first expert and the decision-maker,
i.e., $\beta u(a,\omega)+(1-\beta)v(a,\omega)$ for
some $\beta\in[0,1]$, then we can still construct a truthful equilibrium, where  the first expert continues to obtain his commitment value. 
To see this, let $\alpha_{s}$ be such that
$u(\alpha_{s},\pi_{s}^{*})=\overline{u}(\pi_{s}^{*})$ and note that
$v(\alpha_{p},\pi_{s}^{*})\leq v(\alpha_{s},\pi_{s}^{*})=\max_{\tilde{\alpha}}v(\tilde{\alpha},\pi_{s}^{*})$
for all $s$, where $\alpha_{p}$ is the punishment, which exists  by Lemma 1. This implies that $\beta u(\alpha_p,\pi_{s}^{*})+(1-\beta)v(\alpha_p,\pi_{s}^{*})\leq\beta u(\alpha_{s},\pi_{s}^{*})+(1-\beta)v(\alpha_{s},\pi_{s}^{*})$
for all $\pi_{s}^{*}\in\Pi^{*}$, i.e., $\alpha_{p}$ is also a uniform
punishment for the second expert. We illustrate this remark with a simple example. As in Crawford and Sobel (1982), assume that the decision-maker obtains the payoff $-(\alpha-\omega)^2$, when he chooses $\alpha \in [0,1]$ and the state is $\omega$.\footnote{Throughout, we have assumed that the decision-maker has a finite set of actions. Our results extend to the set $A$ being a non-empty compact subset of $\mathbb{R}$ and concave continuous payoff functions. The proof of Lemma 1 only requires a slight modification: we need to invoke duality for convex programming rather than for linear programming.}
 The payoff of the two experts are $-(\alpha-\omega-b)^2$ and $-(\alpha-\omega-\beta b)^2$, with $\beta \in [0,1]$ and $b >0$, respectively.  The second expert is (weakly) less biased than the first expert. Observe that, up to a constant, the payoff of the less biased expert is a convex combination of the payoff of the most biased expert and the decision-maker, that is:
\begin{equation*}
- \left[(1-\beta)(\alpha- \omega)^2 + \beta (\alpha-\omega-b)^2\right]  =  -(\alpha - \omega - \beta b)^2 -b^2\beta(1-\beta). 
\end{equation*} Therefore, there exists an equilibrium, which gives the most biased expert his commitment value.\footnote{In the quadratic example, the payoff $\overline{u}(\pi)$  to the most biased expert is $-(\mathbb{V}_{\pi}[\omega] + b^2)$, with $\mathbb{V}_{\pi}[\omega]$ the variance of $\omega$ with respect to the distribution $\pi$. Since the variance of a real-valued random variable is concave in its distribution, full information disclosure attains the commitment value.}
\end{rem}
\medskip

\begin{rem}
We have assumed that the choice of experiments is publicly observed.
If the decision-maker does not observe the experiments chosen by the two experts,
but if the experts observe each other's experiment choice, then
again there is a truthful equilibrium, where the optimal  experiment $\sigma^{*}$ is chosen
as in Theorem 1. In this equilibrium, the play on the equilibrium path unfolds
as in Theorem 1. If any expert deviates and chooses another experiment
$\sigma \neq \sigma^{*}$, then the two experts send the message $m_{0}$, where $m_0$
is a message that is never sent on the equilibrium path. If the decision-maker
observes two messages that do not match or observes a message equal
to $m_{0}$ from either of the two experts, then he best responds
to her belief $\pi_{p}\in\text{co }\Pi^{*}$ and plays action $\alpha_{p}$.
\end{rem}

\begin{rem} Similarly, if we assume that the experts do not observe each other's choice of experiments, but the decision-maker does, then our result continues to hold. To see this, we construct an equilibrium as follows. In the first stage, the experts choose the optimal experiment. In the second stage, an expert truthfully reports his signal if he has chosen the optimal experiment in the first stage. (The strategies are left unspecified in other contingencies.) If the decision-maker observes the experts choosing the optimal experiment, the decision-maker follows the same strategy as in our main proof. If the decision-maker observes only one expert choosing the optimal experiment, he plays a best-reply to the message sent by that expert. (The strategies are left unspecified in all other contingencies.) On path, the experts receive their commitment value. If an expert chooses another experiment, the decision-maker observes the deviation but not the other expert, who continues to truthfully reveal the signal. Hence, the deviation does not change the expert's payoff. \end{rem}

\begin{rem}
We have assumed that the two experts choose experiments simultaneously. 
This assumption is again not required for our result.
Suppose instead that one expert, say the first expert, chooses an experiment
$\sigma:\Omega\rightarrow\Delta(S_1\times S_2)$, with expert $i$ privately observing the signal's realization $s_i$.  
As before, after observing their signals, the experts send messages to the decision-maker, who then chooses an action. 
Yet again, we have a truthful equilibrium, where the equilibrium payoff of the two experts is $\text{cav\,}\overline{u}(\pi^{\circ})$ as in Theorem 1. 
In this equilibrium, the first expert chooses the optimal experiment and perfectly correlates the second expert's signal with his own. 
\end{rem}

\begin{rem} Our result relies on the assumption that the two experts send messages simultaneously. Instead, consider a game where in the first stage the two experts independently and simultaneously choose experiments and privately observe signals. In the second stage, the decision-maker consults expert 1 and after observing expert 1's message, the decision-maker sends a cheap-talk message to expert 2. In the third stage, expert 2 sends a message to the decision-maker, and in the last stage the decision-maker chooses an action. Again, there is an equilibrium that delivers the experts their commitment payoff under this specification also. In this equilibrium, both experts choose the expert-optimal experiment, expert 1 truthfully reveals his information, the decision-maker babbles after observing expert 1's message, and expert 2 also truthfully reveals his information. Deviations from equilibrium play are punished by the same mechanism as in our main result. 
\end{rem}

\begin{rem} We have assumed weak perfect Bayesian equilibrium as our solution concept. 
If we restrict attention to a finite set of experiments, which contains $\sigma^*$, then 
we can strengthen the solution
concept to sequential equilibrium. We only need a slight modification of Lemma 1 to guarantee that 
the decision-maker believes that the realized signal is either $s$ or $s'$ after observing a 
report $(s,s')$.  We need to prove the existence of a belief
$\pi_{s,s'}\in\Delta(\{\pi_{s}^{*},\pi_{s'}^{*}\})$ and a mixed action
$\alpha_{s,s'} \in BR(\pi_{s,s'})$ such that $u(\alpha_{s,s'},\pi_{\tilde{s}}^{*})-\overline{u}(\pi_{\tilde{s}}^{*})\leq0$
for all $\tilde{s}\in\{s,s'\}$. A minor adaptation of the proof of Lemma
1 guarantees this result. 
\end{rem}

\begin{proof}[ Proof of Lemma 1.]
We first establish two intermediate claims, then we use these two
claims to establish the lemma. Let $(\lambda_s^*,\pi_s^*)_{s \in S}$ be an optimal splitting. Recall that 
$\Pi^*:=\{\pi_s^*: s\in S\}$ and $\Delta^*$ is the set of all probability distributions over $\Pi^*$.

\medskip 

\textbf{Claim 1:} For any $\lambda\in\Delta^{*}$, $\overline{u}(\sum_{s}\lambda_s\pi_{s}^{*})\leq\sum_{s}\lambda_s\overline{u}(\pi_{s}^{*})$.
\medskip

\textit{Proof of Claim 1:} Consider the convex hull of the graph of $\overline{u}$,
i.e., $\co\{(\pi,r)\in\Delta(\Omega)\times\mathbb{R}:r=\overline{u}(\pi)\}$.
By construction,  the point $(\pi^{\circ},\text{cav}\;\overline{u}(\pi^{\circ}))=(\sum_{s}\lambda^{*}_s\pi_{s}^{*},\sum_{s}\lambda^{*}_s\overline{u}(\pi_{s}^{*}))$
is on the boundary of the convex hull. From the supporting
hyperplane theorem, there exists a hyperplane $h\in\mathbb{R}^{|\Omega|}\times\mathbb{R}$
supporting $\co\{(\pi,r)\in\Delta(\Omega)\times\mathbb{R}:r=\overline{u}(\pi)\}$
at $(\pi^{\circ},\text{cav}\;\overline{u}(\pi^{\circ}))$ such that the graph of $\overline{u}$ lies below $h$. For all $s \in S$, the point
$(\pi_{s}^{*},\overline{u}(\pi_{s}^{*}))$ also lies on the hyperplane $h$.
Consequently, the point $(\sum_s\lambda_s \pi^*_s,\sum_{s}\lambda_s\overline{u}(\pi_{s}^{*}))$,
must also lies on the hyperplane. Therefore, $\overline{u}(\sum_s\lambda_s \pi^*_s)\leq\sum_{s}\lambda_s\overline{u}(\pi_{s}^{*})$
as required. \hfill $\blacksquare$ 

\medskip

\textbf{Claim 2:} Choose any non-empty subset $B \subset A$ and $\varepsilon>0$. If $\max_{s\in S}[u(\alpha,\pi_{s}^{*})-\overline{u}(\pi_{s}^{*})]\geq\varepsilon$
for each $\alpha\in\Delta(B)$, then there exists $\hat{\lambda}\in\Delta^{*}$
such that \newline
$\min_{\alpha\in\Delta(B)}u(\alpha,\sum_{s}\hat{\lambda}_s\pi_{s}^{*}) \geq \sum_{s}\hat{\lambda}_s\overline{u}(\pi_{s}^{*})+\varepsilon$.
\medskip 

\textit{Proof of Claim 2:} The claim follows from duality. Consider the following linear program:
\begin{align*}
\min_{\left(x,\alpha\right)\in\mathbb{R}\times\Delta\left(B\right)}x
\end{align*}
subject to: for all $s \in S$, 
\begin{equation*}
\sum_{a\in B}\alpha(a)\left[u(a,\pi_{s}^{*})-\overline{u}(\pi_{s}^{*})\right]\leq x.
\end{equation*}
This minimization problem has a solution $\hat{x}$. Our hypothesis implies that $\hat{x}\geq\varepsilon$.
The dual program is given by
\[
\max_{\left(y,\lambda\right)\in\mathbb{R}\times\Delta(\Pi^{*})}y
\]
subject to: for all $a \in B$, 
\[ \sum_{s\in S}\lambda_s\left[u(a,\pi_{s}^{*})-\overline{u}(\pi_{s}^{*})\right] \geq y.\]

Since the primal linear program has a solution, the dual program also
has a solution $(\hat{y},\hat{\lambda})$. No duality gap further
implies that $\hat{y}=\hat{x}\geq\varepsilon$. (See Section 4.2 of \citealp{Luen08}.) Therefore, for all $a \in B$, 
\[
u(a,\sum_{s}\hat{\lambda}_s\pi_{s}^{*})=\sum_{s\in S}\hat{\lambda}_s u(a,\pi_{s}^{*})\geq\varepsilon+\sum_{s\in S}\hat{\lambda}_s\overline{u}(\pi_{s}^{*})
\]
Hence, $u(\alpha,\sum_{s}\hat{\lambda}_s\pi_{s}^{*})\geq\sum_{s}\hat{\lambda}_s\overline{u}(\pi_{s}^{*})+\varepsilon$
for all $\alpha\in\Delta(B)$, as required. \hfill $\blacksquare$

\medskip 

We now use Claims 1 and 2 to complete the proof. Denote by $br(\pi)\subset A$
the decision-maker's set of pure best-replies to belief $\pi$. \medskip 

By contradiction, assume that there does not exist $\pi_{p}\in \mathrm{co}\, \Pi^{*}$ and $\alpha_{p}\in BR(\pi_{p})$
such that $u(\alpha_{p},\pi_{s}^{*})-\overline{u}(\pi_{s}^{*})\leq 0$
for all $\pi_{s}^{*}\in\Pi^{*}$. Note that $\pi\in \mathrm{co}\, \Pi^{*}$ if
and only if $\pi=\sum_{s}\lambda_{s}\pi_{s}^{*}$ for some $\lambda\in\Delta^{*}$.
Hence, our contradiction hypothesis can be restated as follows: for
each $\lambda\in\Delta^*$, there exists $\varepsilon(\lambda)>0$
such that $\max_{s\in S}[u(\alpha,\pi_{s}^{*})-\overline{u}(\pi_{s}^{*})]\geq\varepsilon(\lambda)$
for each $\alpha\in\Delta(br(\sum_{s}\lambda_{s}\pi_{s}^{*}))=BR(\sum_{s}\lambda_{s}\pi_{s}^{*})$.
Let $\varepsilon:=\min_{\lambda\in\Delta^*}\varepsilon(\lambda)$.
Note that $\varepsilon>0$ because $\varepsilon(\lambda)$ depends only
on the finite set $br(\sum_{s}\lambda_{s}\pi_{s}^{*})$, and there
are finitely many such subsets of $A$.

Define the correspondence $F:\Delta^{*}\rightarrow\Delta^{*}$, with
\[
F(\lambda):=\Big\{\lambda'\in\Delta^*:\min_{\alpha\in BR(\sum_{s}\lambda_{s}\pi_{s}^{*})}\sum_{s}\lambda'_s\Big(u(\alpha,\pi_{s}^{*})-\overline{u}(\pi_{s}^{*})\Big)\geq\varepsilon\Big\}.
\]
We can readily check that this correspondence is convex and
compact valued. We argue below that it is non-empty valued
and lower hemi-continuous. Hence, the correspondence has a fixed point
$\overline{\lambda}\in F(\overline{\lambda})$ by Theorem 15.4 in
\citet{border1990fixed}. Noting that $\sum_{s}\overline{\lambda}(s)u(\alpha,\pi_{s}^{*})=u(\alpha,\sum_{s}\overline{\lambda}(s)\pi_{s}^{*})$,
we find
\[
\min_{\alpha\in BR(\sum_{s}\overline{\lambda}(s)\pi_{s}^{*})}\Big(u(\alpha,\sum_{s}\overline{\lambda}(s)\pi_{s}^{*})-\sum_{s}\overline{\lambda}(s)\overline{u}(\pi_{s}^{*})\Big)\geq\varepsilon
\]
for $\overline{\lambda}\in\Delta^{*}$ contradicting Claim 1 and establishing
the result.
\medskip

We now show that the correspondence is non-empty valued. Pick any $\lambda\in\Delta^{*}$. The contradiction
hypothesis states that $\max_{s\in S}[u(\alpha,\pi_{s}^{*})-\overline{u}(\pi_{s}^{*})]\geq\varepsilon$
for each $\alpha\in BR(\sum_{s}\lambda_s\pi_{s}^{*})$. Claim
2 then implies that there exists $\hat{\lambda}\in\Delta^{*}$ such that
\[\min_{\alpha\in BR(\sum_{s}\lambda_s\pi_{s}^{*})}\sum_{s}\hat{\lambda}_s(u(\alpha,\pi_{s}^{*})-\overline{u}(\pi_{s}^{*}))\geq\varepsilon,\]
i.e., the correspondence is non-empty valued. 
\medskip 

Finally, we prove lower hemi-continuity.  Pick an open set $O\subseteq\Delta^{*}$
such that $F(\lambda)\cap O\neq\emptyset$. Since $BR$ is upper hemi-continuous
(by the maximum principle) and $A$ is finite, there exists a neighborhood
$O'$ of $\lambda$ such that $BR(\sum_{s}\lambda'_s\pi_{s}^{*})\subseteq BR(\sum_{s}\lambda_s\pi_{s}^{*})$
for all $\lambda' \in O'$. Therefore, for all  $\lambda' \in O'$,
\[
\min_{\alpha\in BR(\sum_{s}\lambda'_s \pi_{s}^{*})}\sum_{s}\lambda^{''}_s\Big(u(\alpha,\pi_{s}^{*})-\overline{u}(\pi_{s}^{*})\Big)\geq\min_{\alpha\in BR(\sum_{s}\lambda_s \pi_{s}^{*})}\sum_{s}\lambda^{''}_{s}\Big(u(\alpha,\pi_{s}^{*})-\overline{u}(\pi_{s}^{*})\Big)\geq\varepsilon
\]
for any $\lambda^{''} \in F(\lambda)\cap O$ because $BR(\sum_{s}\lambda'_s\pi_{s}^{*})\subseteq BR(\sum_{s}\lambda_s\pi_{s}^{*})$,
i.e., $\lambda^{''}\in F(\lambda')$. 
Hence, $F(\lambda')\cap O\neq\emptyset$ for all $\lambda' \in O'$,
which proves the lower hemi-continuity of $F$ (Definition 11.3 in
\citet{border1990fixed}).
\end{proof}

\section{The Decision-maker and Cross-verification}

The previous section showed that the experts benefit from the decision-maker cross-verifying their information. 
This section explores whether the decision-maker can also benefit from cross-verification. 
\medskip

We begin with some definitions. Fix a cheap-talk game $\Gamma(\pi^{\circ},u,v)$. We say that 
the experts benefit from persuasion if $\text{cav\,}\overline{u}(\pi^{\circ})>\overline{u}(\pi^{\circ})$. Similarly, we say that the decision-maker benefits from cross-verification if there exists an equilibrium of the cheap-talk game, where the decision-maker's payoff exceeds the ex-ante payoff $\max_{a\in A}v(a,\pi^{\circ})$. Notice that if the decision-maker benefits from cross-verification, the experts must reveal some information to the decision-maker. \medskip 

Define $\hat{A}:=\{a\in A:\exists\pi\in\Delta(\Omega)\text{ s.t. }a\in BR(\pi)\}$
and $\overline{v}(\pi):=\max_{\alpha}v(\alpha,\pi)$ for all $\pi\in\Delta(\Omega)$.
We say that there are \emph{no redundant actions} for the decision-maker if for all non-empty $B\subset\hat{A}$,
there exists $\pi\in\Delta(\Omega)$ such that $\overline{v}(\pi)>\max_{a\in B}v(a,\pi)$.
There are no redundant actions for the experts if there are no two distinct
actions $a$ and $a'$ such that $u(a,\omega)=u(a',\omega)$ for all
$\omega\in\Omega$. \medskip

\begin{rem}
The conditions of non-redundancy are generic. Moreover, the condition of no redundant actions for the decision-maker does not
preclude strictly dominated actions. Two important implications of that condition are as follows: (i) the set $BR^{-1}(a):=\{\pi\in\Delta(\Omega):v(a,\pi)=\overline{v}(\pi)\}$
has full dimension  (as a subset of the simplex of dimension $|\Omega|-1$), and (ii) no action other than $a$ is optimal in the
relative interior of $BR^{-1}(a)$, denoted by $\mathrm{int\,}BR^{-1}(a)$.
\end{rem}
\medskip

Theorem 1 showed that the experts benefit from cross-verification
in games where they benefit from persuasion. The following proposition
further establishes that the decision-maker also benefits from cross-verification in
such games.

\begin{prop}
Assume that there are no redundant actions for the decision-maker in the game
$\Gamma(\pi^{\circ},u,v)$. At almost all priors $\pi^{\circ}$, if the experts benefit from persuasion, then the decision-maker benefits from cross-verification.
\end{prop}

We first illustrate the logic of the proposition with the help of a simple example. There are two states, $\omega_0$ and $\omega_1$, and three actions, $a_0,a_1$ and $a_p$. Throughout the example, probabilities refer to the probability of $\omega_1$. The prior is $\pi^{\circ}=0.45$. The payoffs are illustrated in Figure \ref{Fig:example pure strategy}. The optimal experiment consists in splitting the prior into the posteriors $\pi_{s_{0}}^{*}=0.3$ and $\pi_{s_{1}}^{*}=0.6$. The experts strictly benefit from persuasion. From Theorem 1, there exists a truthful equilibrium, where an expert's payoff is his commitment value.  Action $a_p$ is the uniform punishment sustaining the equilibrium. Note that $a_p$ is uniquely optimal at the prior and also optimal at the two posteriors. Consequently,  the decision-maker does not benefit from cross-verification at the equilibrium.  Yet, we can construct another equilibrium, where the decision-maker benefits from cross-verification. To see this, consider the splitting of the prior into $\pi_{s_0}=0.2$ and $\pi_{s_1}=0.8$. At $\pi_{s_0}$ (resp., $\pi_{s_1}$), the decision-maker plays $a_0$ (resp., $a_1$). To sustain this splitting as an equilibrium, the decision-maker punishes the experts with $a_p$.  The decision-maker strictly benefits from this more informative experiment.  \medskip 

We prove that the logic of the example generalizes to almost all priors. That is, for all priors, but for a subset with Lebesgue measure zero, we can always construct an equilibrium of the cheap-talk game, where the decision-maker benefits from cross-verification if the experts benefit from persuasion. More precisely, we prove that the proposition holds at all interior priors, where the decision-maker has at most two best-replies, a generic condition.\medskip

The need for non-redundancy is clear. If the decision-maker is indifferent between all his actions, the decision-maker cannot benefit from cross-verification, while the experts can benefit from persuasion. We now turn to the proof.

\begin{figure}
\subfloat[{DM's Preferences. Action $a_{0}$ is optimal for the DM for $\pi\in[0,0.3]$,
action $a_{p}$ is optimal for $\pi\in[0.3,0.6]$, and $a_{1}$ is
optimal for $[0.6,1]$.}]{\includegraphics[scale=0.9]{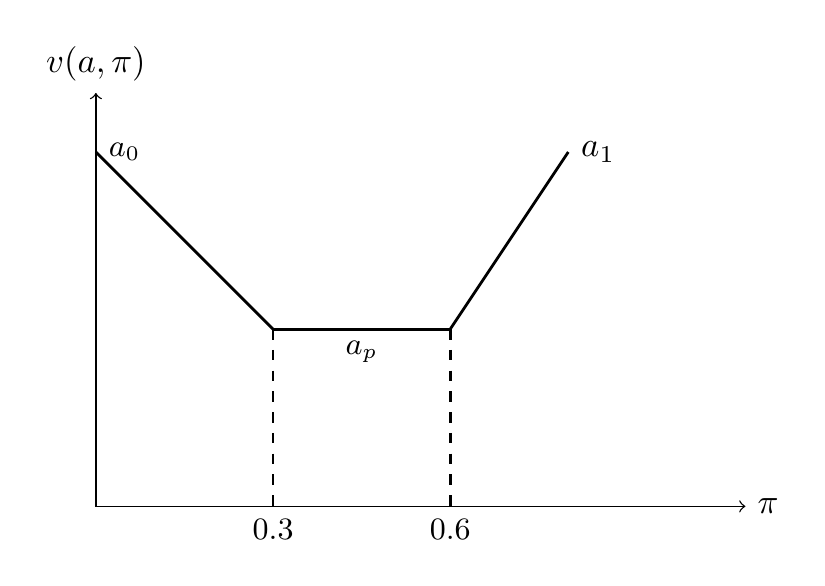}}\qquad{}\subfloat[{Expert's Preferences. The solid black lines depict $\overline{u}(\pi)$,
the solid blue line depicts $\text{cav }\overline{u}(\pi)$ for $\pi\in[0,3,0.6]$,
and $\text{cav }\overline{u}(\pi)$ coincides with $\overline{u}(\pi)$
for $\pi\protect\notin[0,3,0.6]$.}]{\includegraphics[scale=0.9]{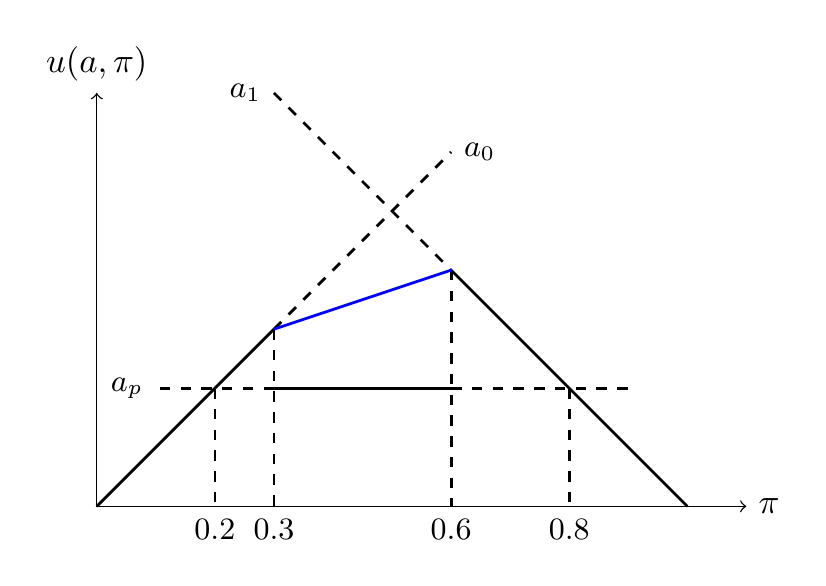}\label{fig:example pure strategy DM}}\caption{DM Benefits from Cross-verification.}
\label{Fig:example pure strategy}
\end{figure}

\begin{proof}[Proof of Proposition 1.]
Consider an optimal splitting $(\lambda^*_s,\pi_s^*)_{s \in S}$ of $\pi^{\circ}$, which induces the value $\cav \overline{u}(\pi^{\circ})$, where $ \cav \overline{u}(\pi^{\circ})> \overline{u}(\pi^{\circ})$. Without loss of generality, assume that $\lambda_{s}^*>0$ for all $s \in S$. Let $\overline{v}(\pi):=\max_{\alpha} v(\alpha,\pi)$ for all $\pi \in \Delta(\Omega)$. \medskip

If the decision-maker benefits from the statistical experiment, there is nothing to prove. So, assume that  the decision-maker does not benefit from the statistical experiment, i.e., $\sum_{s} \lambda^*_s \overline{v}(\pi^*_s)=\overline{v}(\pi^{\circ})$. We construct another equilibrium at which the decision-maker benefits from cross-verification. \medskip  

We first claim that for all $a \in BR(\pi^{\circ})$, $a \in BR(\pi)$ for all $\pi \in \co\{\pi^*_s: s \in S\}$. To see this, consider any $a \in BR(\pi^{\circ})$ and observe that 
\begin{eqnarray*}
\sum_{s} \lambda^*_s\overline{v}(\pi^*_s)=\overline{v}(\pi^{\circ})=v(a,\pi^{\circ})= v\left(a,\sum_{s}\lambda^*_s \pi_s^*\right) = \sum_{s}\lambda^*_s v(a,\pi_s^*).
\end{eqnarray*}
It follows that 
\begin{eqnarray*}
\sum_{s}\underbrace{\lambda^*_s}_{>0}  (\underbrace{\overline{v}(\pi^*_s)-v(a,\pi_s^*)}_{\geq 0})=0
\end{eqnarray*}
If there exists $s$ such that $\overline{v}(\pi_s^*) > v(a,\pi^*_s)$, we have a contradiction. Hence, $a \in BR(\pi_s^*)$ for all $s$ and, consequently, $a \in BR(\pi)$ for all $\pi \in \co\{\pi^*_s: s \in S\}$. 
\medskip 

From the definition of $\overline{u}$, we have that $u(a,\pi_s^*) \leq \overline{u}(\pi_s^*)$ for all $s$, for all $a \in BR(\pi^{\circ})$, since 
$BR(\pi^{\circ}) \subseteq BR(\pi_s^*)$ for all $s$. We now argue that for all $a \in BR(\pi^{\circ})$, there exists $s_{a} \in S$ such that 
$u(a,\pi_{s_{a}}^*) < \overline{u}(\pi_{s_{a}}^*)$. Choose any $a \in BR(\pi^{\circ})$. To the contrary, assume that $u(a,\pi_s^*) = \overline{u}(\pi_s^*)$ for all $s$. We then have
\begin{eqnarray*}
\cav \overline{u}(\pi^{\circ})=  \sum_{s} \lambda_s^* \overline{u}(\pi_s^*)= \sum_{s} \lambda_s^* u(a,\pi_s^*)
= u(a,\pi^{\circ}) \leq \overline{u}(\pi^{\circ}) \leq \cav \overline{u}(\pi^{\circ}),
\end{eqnarray*}
a contradiction with the expert benefiting from the experiment. \medskip

To sum up, we have (i) $BR(\pi^{\circ}) \subseteq BR(\pi)$ for all  $\pi \in \co\{\pi^*_s: s \in S\}$, and (ii) for each $a \in BR(\pi^{\circ})$, there exists $s_{a}$ such that $u(a^*_{s_{a}},\pi_{s_{a}}^*) > u(a,\pi_{s_{a}}^*)$ with $a^*_{s_{a}} \in BR(\pi_{s_{a}}^*)$ satisfying $u(a^*_{s_{a}},\pi_{s_{a}}^*)=\overline{u}(\pi_{s_{a}}^*)$.\medskip

For each $a \in BR(\pi^{\circ})$, consider the open ball $\mathcal{O} =\{\pi \in \Delta(\Omega): ||\pi-\pi^*_{s_{a}}|| < \varepsilon\}$ such that  $u(a,\pi) <  u(a^*_{s_{a}},\pi)$ for all $\pi$ in the open ball. Since $u$ is continuous in $\pi$ and $u(a,\pi_{s_{a}}^*) <  u(a^*_{s_{a}},\pi_{s_{a}}^*)$, such an open ball exists. \medskip 

We claim that $\mathcal{O}$ intersects the relative interior of $BR^{-1}(a^*_{s_{a}})$. To see this, note that $\mathcal{O} \cap BR^{-1}(a^*_{s_{a}}) \neq \emptyset$ since $\pi^*_{s_{a}}$ is an element of both $\mathcal{O}$ and $BR^{-1}(a^*_{s_{a}})$. Moreover, it follows from the non-redundancy of $A$ that $\pi^*_{s_{a}}$ is not in the relative interior of $BR^{-1}(a^*_{s_{a}})$ since  any $a \in BR(\pi^{\circ})$ is also optimal at $\pi^*_{s_{a}}$. Since the relative interior of $BR^{-1}(a^*_{s_{a}})$ is non-empty, there exists $\pi^{**}$ in the relative interior such that the half-open line segment $[\pi^{**},\pi^*_{s_{a}})$ is contained in the relative interior. (See Theorem 2.1.3 and Lemma 2.1.6 in Hiriart-Urruty and Lemar\'echal.) Therefore, there exists $\overline{\pi}_a$ in the intersection of the relative interior of $BR^{-1}(a^*_{s_{a}})$ and $\mathcal{O}$, i.e., such that $u(a,\overline{\pi}_a)< u(a_{s}^*,\overline{\pi}_a)=\overline{u}(\overline{\pi}_a)$. Note that $v(a_{s_{a}}^*,\overline{\pi}_a) > v(a,\overline{\pi}_a)$ since $a_{s_{a}}^*$ is uniquely optimal at $\overline{\pi}$. In other words, there is an element of $BR(\pi^{\circ})$, namely $a$, which is not an element of $BR(\overline{\pi}_a)$.\medskip 

The last step consists in showing that there exists $a \in BR(\pi^{\circ})$ and $\underline{\pi}_a  \in BR^{-1}(a)$ such that  the open segment $(\underline{\pi}_a,\overline{\pi}_a)$ includes $\pi^{\circ}$. Indeed, if such an open segment exists, we have a splitting $(\underline{\pi}_a,\overline{\pi}_a)$ of $\pi^{\circ}$ such that $\overline{u}(\underline{\pi}_a) \geq u(a,\underline{\pi}_a)$, $\overline{u}(\overline{\pi}_a)=u(a_{s_{a}}^*,\overline{\pi}_a) > u(a,\overline{\pi}_a)$.  This splitting can be supported as a truthful equilibrium (with $a$ as the punishment at belief $\pi^{\circ}$). Moreover, since $v(a_{s_{a}}^*,\overline{\pi}_a) > v(a,\overline{\pi}_a)$, the decision-maker strictly benefits, the desired contradiction. \medskip

Finally, suppose that $\pi^{\circ}$ is in the interior of the simplex. If $BR(\pi^{\circ})=\{a\}$, then  $\pi^{\circ}$ is in the relative interior of $BR^{-1}(a)$. Thus, we can trivially find a segment with the required property. \medskip

If $BR(\pi^{\circ})=\{a,b\}$ and $s_a=s_b$, then the same arguments apply, since the open segment will intersect either 
$BR^{-1}(a)$ or $BR^{-1}(b)$. If $s_a \neq s_b$, choose $\overline{\pi}_{s_a}$ such that $b$ is uniquely optimal at $\overline{\pi}_{s_a}$.  Such $\overline{\pi}_{a}$  exists since $v(b,\pi_{s_a})= \max_{a' \in BR(\pi_{s_a})}v(a',\pi_{s_a})$ (if not $s_a=s_b$). As before, the open segment intersects either 
$BR^{-1}(a)$ or $BR^{-1}(b)$. However, it cannot be $BR^{-1}(b)$. If it were, $b$ would be uniquely optimal at $\overline{\pi}_{s_a}$ and optimal at $\pi^{\circ}$ and $\underline{\pi}_{s_a}$, which is not possible since $BR^{-1}(b)$ is convex. \medskip 

Since the set of interior priors with at most two best-replies is generic, the proof is complete.\end{proof}

Proposition 1 does not generalize to all priors. For a counter-example, consider  Figure \ref{fig:prop1}. There are three states, $\omega_0$, $\omega_1$ and $\omega_2$, and two actions, $a$ and $b$. The action $a$ (resp., $b$) is optimal in the left triangle marked ``$a$'' (resp., in the right triangle marked ``$b$''). At the prior $\pi^{\circ}$, the action $a$ is the unique best-reply of the decision-maker. Assume that $u(b,\omega_1) >u(a,\omega_1)$.  Thus, if the experts truthfully reveal the state, they benefit from persuasion, while the decision-maker does not.\footnote{If there are two states, Proposition 1 generalizes to all interior priors. In this case, non-redundancy of the decision-maker's payoff implies that the decision-maker has at most two best-replies at each belief, where we know that Proposition 1 holds. In general, however, we do not know whether the proposition generalizes to all interior priors.}

\begin{figure}[h]
\begin{center}
\begin{tikzpicture}[scale=1.1]
       \draw[fill=lightgray] (0,0)--(4,0)  -- (2,3.46) -- (0,0);
    \node[below] at (0,0) {$\omega_0$};
     \node[above] at (2,3.46) {$\omega_1$}; 
     \node[below] at (4,0) {$\omega_2$};
          \draw[fill] (1,1.7) circle [radius=0.05];
     \node[left] at (1,1.7) {$\pi^{\circ}$};
\draw[dashed] (2,0)--(2,3.46);
\node[] at (1.25,1) {$a$};
\node[] at (2.75,1) {$b$};
   \end{tikzpicture}
 \end{center}
\caption{A counter-example}
\label{fig:prop1}
\end{figure}
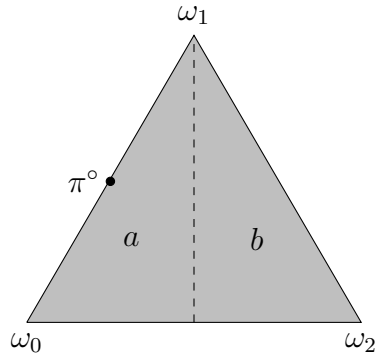


\medskip

Proposition 1 proved that the decision-maker benefits from cross-verification
whenever the experts benefit from persuasion. We now show a partial converse, that is, 
the decision-maker benefits from cross-verification only when the experts benefit from persuasion.

\begin{prop}
Assume that there are no redundant actions for the experts and the
decision-maker in the game $\Gamma(\pi^{\circ},u,v)$. If $\overline{u}$ is a
concave function, then the decision-maker does not benefit from cross-verification. That is, in all equilibria of  $\Gamma(\pi^{\circ},u,v)$, 
the decision-maker's  payoff is $\overline{v}(\pi^{\circ})$. 
\end{prop}

To understand Proposition 2, assume that the experts and the decision-maker have opposing preferences, that is, $u=-v$. In this case, what is best for the decision-maker is worst for the experts, and therefore, $\overline{v}=-\overline{u}$. Moreover, if either of the experts, say expert 1, chooses a uninformative experiment,  an expert's payoff is $u(\pi^{\circ})$ in all equilibria of the ensuing game. This is because if expert 2's experiment produces two signals $s$ and $s'$ such that the set of best-replies at $\pi_s$ differs from the set of best-replies at $s'$, then expert 2 has an incentive to misreport one of the two signals, if not both. The experts cannot credibly communicate any information. Therefore, no expert can obtain less than $\overline{u}(\pi^{\circ})$ in equilibrium. Experts cannot obtain more than $\overline{u}(\pi^{\circ})$ either. Indeed, for every on-path posterior $\pi$, the decision-maker chooses a best-reply in equilibrium, hence an expert's payoff is minimized at $\pi$, i.e., an expert's payoff  is $u^{\min}(\pi):=\min_{a}u(a,\pi)$. The result then follows from the concavity of $u^{\min}$.   Proposition 2 does not require opposing preferences; the logic outlined above extends to all games, where $\overline{u}$ is concave.\medskip

To further illustrate Proposition 2, consider Figure \ref{Fig: no communication example.}. 
For the decision-maker to benefit from cross-verification, the experts would need to choose an experiment, which induces the decision-maker to play different actions after receiving different signals.  However, we cannot sustain such a choice as an equilibrium. An expert would always have an incentive to misreport the realized signal. This is because any action other than the one chosen by the decision-maker improves an expert's payoff, i.e., there is no uniform punishment. \medskip 

The need for the non-redundancy of the experts' actions is again clear. If the experts are totally indifferent, they cannot benefit from persuasion but can provide the decision-maker with perfectly informative signals.  It remains to prove Proposition 2. We do so through a series of lemmata. The following lemma shows that the conflict of interest between the experts and the decision-maker is maximal when the experts cannot benefit from persuasion; that is, the decision-maker's best-replies at belief $\pi$ minimizes the experts' expected payoff. Recall that $\hat{A}$ is the set of actions that are a best response for the decision-maker to some belief.

\begin{figure}
\subfloat[{DM's Preferences. Action $a_{0}$ is optimal for the DM for $\pi\in[0,0.4]$
and $a_{1}$ is optimal for $[0.4,1]$.}]{\includegraphics[scale=0.9]{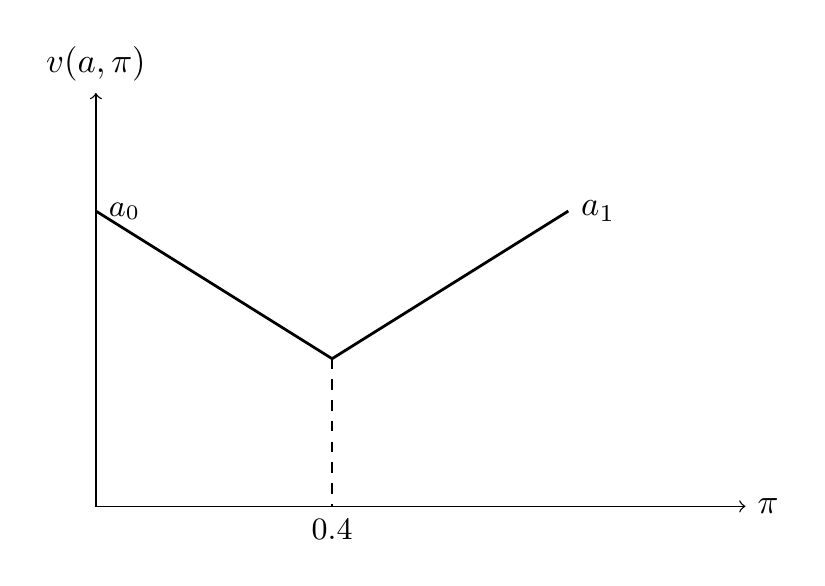}}\qquad{}\subfloat[Expert's Preferences. The solid black lines depict the concave function
$\overline{u}(\pi)$.]{\includegraphics[scale=0.9]{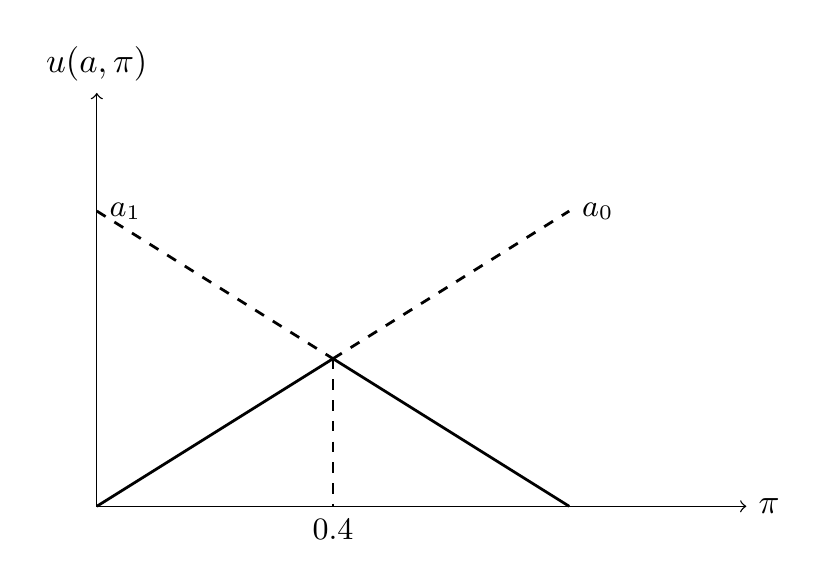}}\caption{No Benefit from Cross-verification or Persuasion.}
\label{Fig: no communication example.}
\end{figure}

\begin{lem}
\label{L:conflicteverywhere}For every $\pi\in\Delta(\Omega)$, $BR(\pi)=\arg\min_{\alpha'\in\Delta(\hat{A})}u(\alpha',\pi)$.
\end{lem}
\begin{proof}[Proof of Lemma \ref{L:conflicteverywhere}] We start by proving the following claim. 
\begin{claim}
\label{L:conflictinterest}$a\in \hat{A}$ and $\pi\in \mathrm{int\,} BR^{-1}(a)$ implies
$\{a\}=\arg\min_{\alpha'\in\Delta(\hat{A})}u(\alpha',\pi)$.
\end{claim}
\begin{proof}[Proof of Claim \ref{L:conflictinterest}] Fix $a\in \hat{A}$ and $\pi\in \mathrm{int\,} BR^{-1}(a)$. We first argue that there 
does not exist $a' \in \hat{A}$ such that  such that $u(a',\pi)<u(a,\pi)$. To the contrary, suppose such $a'$ exists. Pick an arbitrary $\pi'\in \mathrm{int\,} BR^{-1}(a')$. There
exists $\pi''\in \mathrm{int\,} BR^{-1}(a')$ and $\lambda\in(0,1)$ such that
$\pi''=\lambda\pi+(1-\lambda)\pi'$. We obtain 
\begin{eqnarray*}
u(a',\pi'') & = & \overline{u}(\pi'') \\
&\geq & \lambda \overline{u}(\pi)+ (1-\lambda) \overline{u}(\pi')\\
& = & \lambda u(a,\pi)+(1-\lambda)u(a',\pi')\\
 &> & \lambda u(a',\pi)+(1-\lambda)u(a',\pi')\\
 &=& u(a',\pi''),
\end{eqnarray*}
where the first inequality follows from the concavity of $\overline{u}$, the desired contradiction.\medskip  

We now argue that there does not exist $a' \in \hat{A}$ such that  $u(a',\pi)=u(a,\pi)$. From the above, for all $\pi_n \in \mathrm{int\,} BR^{-1}(a)$, 
$u(a',\pi_n)\geq u(a,\pi_n)$. Consider any convex combination $(\lambda_n,\pi_n)_n$  satisfying $\sum_n \lambda_n \pi_n =\pi$, $\pi_n \in \mathrm{int\,} BR^{-1}(a)$ for all $n$, $\lambda_n>0$ for all $n$, and the $\pi_n$ being linearly independent.  Such a convex combination exists since $BR^{-1}(a)$ has full dimension. If $u(a',\pi)=u(a,\pi)$, then 
\[ u(a',\pi)= \sum_{n} \lambda_n u(a',\pi_n) \geq \sum_n \lambda_n u(a,\pi_n)=u(a,\pi) = u(a',\pi), \]
i.e.,  $u(a',\pi_n)=u(a,\pi_n)$ for all $n$,  a
contradiction with the condition of no redundant actions for the experts. Therefore, for all $a' \neq a$, $u(a',\pi)>u(a,\pi)$, which completes the proof of the claim.
\end{proof}

From Claim \ref{L:conflictinterest}, the statement is true for all $\pi$ such that 
 $\pi\in \mathrm{int\,} BR^{-1}(a)$ for some $a \in \hat{A}$.
Since $BR$ and $\arg\min_{\alpha'\in \hat{A}}u(\alpha',\pi)$ are upper
hemi-continuous correspondences, which coincide almost everywhere
(in Lebesque measure), they coincide everywhere.\end{proof}

We now derive an immediate implication of Lemma \ref{L:conflicteverywhere}. We first introduce some additional notation. Recall that following the choice of experiments $(\sigma_1,\sigma_2)$, we have a proper sub-game. We are interested in analyzing the play in these sub-games. To ease notation, we drop the dependence on 
$(\sigma_1,\sigma_2)$ and write $\pi(m_{1},m_{2})\in\Delta(\Omega)$ for the decision-maker's belief after observing the messages $(m_1,m_2)$. Similarly, we write $\alpha(m_{1},m_{2})$ for the decision-maker's equilibrium
reply. Notice that $\alpha(m_{1},m_{2})\in \Delta(\hat{A})$ because this action is a  best response to belief  $\pi(m_{1},m_{2})$. Finally, let $\mathbb{P}$
denote the probability distribution over signals, messages and actions induced by the prior and the strategy profile, conditional on the experiments 
$(\sigma_1,\sigma_2)$. At an equilibrium, sequential rationality requires the decision-maker to choose a best-reply to his belief.  Fix an equilibrium,  an on-path profile of messages $(m_1,m_2)$, and its associated belief $\pi(m_1,m_2)$. Since all best-replies of the decision-maker to $\pi(m_1,m_2)$ minimize the experts' payoffs, no expert must be able to induce the decision-maker to choose an action outside $BR(\pi(m_1,m_2))$ by changing his message to $m'_1$. 

\begin{lem}
\label{independencefrommessages}If $\mathbb{P}(m_{i},m_{j})>0$,
then for all $m_{i}'$, $\alpha(m_{i}',m_{j})\in BR(\pi(m_{i},m_{j}))$.
\end{lem}
\begin{proof}[Proof of Lemma \ref{independencefrommessages}]
Without loss of generality, let $i=1$, $j=2$. The proof is by contradiction. Assume that there exists
$m_{1},m_{1}', m_{2}'$ such that  $\alpha(m_{1}',m_{2}') \notin BR(\pi(m_{1},m_{2}'))$. 

From Lemma \ref{L:conflicteverywhere}, $u(\alpha(m_{1}',m_{2}'),\pi(m_{1},m_{2}'))>u(\alpha(m_{1},m_{2}'),\pi(m_{1},m_{2}'))$.
The equilibrium payoff to expert 1 is 
\[
\sum_{(\tilde{m}_{1},\tilde{m}_{2})}\mathbb{P}(\tilde{m}_{1},\tilde{m}_{2})u(\alpha(\tilde{m}_{1},\tilde{m}_{2}),\pi(\tilde{m}_{1},\tilde{m}_{2})).
\]
If expert 1 deviates by always sending the message $m_1'$, his expected payoff is: 
\[
\sum_{(\tilde{m}_{1},\tilde{m}_{2})}\mathbb{P}(\tilde{m}_{1},\tilde{m}_{2})u(\alpha(m'_{1},\tilde{m}_{2}),\pi(\tilde{m}_{1},\tilde{m}_{2})).
\]
We now argue that the deviation is profitable, the required contradiction. 

From Lemma \ref{L:conflicteverywhere}, we have that  $u(\alpha(\tilde{m}_{1},\tilde{m}_{2}),\pi(\tilde{m}_{1},\tilde{m}_{2}))\leq u(\alpha(m'_{1},\tilde{m}_{2}),\pi(\tilde{m}_{1},\tilde{m}_{2}))$
for all $(\tilde{m}_{1},\tilde{m}_{2})$. Moreover, there exists $(m_1,m_2)$ such that the inequality is strict and
$\mathbb{P}(m_1,m_2)>0$. Thus, the deviation is profitable. 
\end{proof}

The next lemma shows that if any  expert chooses an uninformative experiment, then the experts' and the decision-maker's payoff in the ensuing equilibrium is equal to their payoff at their prior belief.

\begin{lem}
\label{L:uninformativeisguaranteed} Let $(\sigma_1,\sigma_2)$ be a profile of experiments. 
If either $\sigma_{1}$ or $\sigma_2$ is an uninformative
experiment, then the experts' equilibrium payoff is $\overline{u}(\pi^{\circ})$ and the decision-maker's equilibrium payoff is 
$\overline{v}(\pi^{\circ})$  in the ensuing sub-game.
\end{lem}
\begin{proof}[Proof of Lemma \ref{L:uninformativeisguaranteed}] 
Without loss of generality, assume that $\sigma_2$ is uninformative. Since the experiments are observed by the decision-maker, 
this implies that $\pi(m_1,m_2)$ is independent of $m_2$. (Recall that we require the beliefs to be consistent with the experiments.)
To ease the notation, we drop the dependence on $m_2$. \medskip

Together with Lemma \ref{independencefrommessages}, this
implies that for all $(m_1,m_2)$ such that $\mathbb{P}(m_{1},m_{2})>0$,  $\alpha(m'_{1},m_{2})\in BR(\pi(m_{1}))$
for all $m'_{1}$. That is, $\alpha(m'_{1},m_{2})$ is a best-reply to all posterior beliefs $\pi(m_1)$. Note that since 
$\mathbb{P}(m_{1},m_{2})>0$, the message $m_1$ has strictly positive probability. It follows that $\alpha(m'_{1},m_{2})$ is a best-reply to $\pi^{\circ}$ (as the prior is a convex combinations of the posteriors). Since it is true for all $(m'_1,m_2)$, the decision-maker payoff is $\overline{v}(\pi^{\circ})$. 
\medskip 

Finally, since Lemma \ref{L:conflicteverywhere}
states that the experts are indifferent among all best-replies of the decision-makers, an expert's payoff is $\overline{u}(\pi^{\circ})$.\end{proof}

We now conclude the proof.

\begin{lem}\label{eq-payoff}
In any equilibrium of the cheap-talk game, the experts' payoff is $\overline{u}(\pi^{\circ})$, and the decision-maker's payoff is 
$\overline{v}(\pi^{\circ})$.
\end{lem}
\begin{proof}[Proof of Lemma \ref{eq-payoff}]
Fix any equilibrium of the cheap-talk game. From 
Lemma \ref{L:uninformativeisguaranteed}, the payoff to any expert must at least be  $\overline{u}(\pi^{\circ})$. We now argue that it cannot be higher. If $(\sigma^*_1,\sigma^*_2)$ are the experiments chosen at the first stage, then in the ensuing sub-game, an expert's payoff is:
\begin{align*}
\sum_{(m_{1},m_{2})}\mathbb{P}(m_{1},m_{2})u(\alpha(m_{1},m_{2}),\pi(m_{1},m_{2})) & =\sum_{(m_{1},m_{2})}\mathbb{P}(m_{1},m_{2})\min_{a\in A}u(a,\pi(m_{1},m_{2}))\\
 & \leq\min_{a\in A}u\left(a,\sum_{(m_{1},m_{2})}\mathbb{P}(m_{1},m_{2})\pi(m_{1},m_{2})\right)\\
 & =\min_{a\in A}u\left(a,\pi^{\circ}\right)=\bar{u}(\pi^{\circ}).
\end{align*}
(Recall that $\mathbb{P}$, $\alpha$ and $\pi$ depend on $(\sigma_1^*,\sigma_2^*)$,  but to ease notation, we do not explicitly write the dependence.)\medskip 

Finally, we argue that the decision-maker cannot get a payoff higher than $\overline{v}(\pi^{\circ})$ either. Indeed, for the decision-maker to obtain a higher payoff, there must exist an action
$a\in BR(\pi^{\circ})$ and a message profile $(m_{1},m_{2})$
such that $\mathbb{P}(m_1,m_2)>0$ and $a\notin BR(\pi(m_{1},m_{2}))$. This, however, would imply
that an expert's equilibrium payoff is strictly less than $u(a,\pi^{\circ})$, a contradiction with an expert's equilibrium payoff being equal to
$\bar{u}(\pi^{\circ})=\min_{a'\in \hat{A}}u\left(a',\pi^{\circ}\right)$. \medskip

The latter assertion follows from Lemma \ref{independencefrommessages}, which states that 
$u(a,\pi(m_{1},m_{2}))>u(\alpha(m_{1},m_{2}),\pi(m_{1},m_{2}))$
and $u(a,\pi(m_{1}',m_{2}'))\geq u(\alpha(m_{1}',m_{2}'),\pi(m_{1}',m_{2}'))$
for all pairs of messages  $(m_{1}',m_{2}')$ with $\mathbb{P}(m_1',m_2')>0$.
\end{proof}

\section{Conclusion}
In this paper, we studied the effects of cross-verification on the decision-maker's and experts' payoffs. Clearly, cross-verification is not the sole reason for soliciting advice from multiple experts. Consulting a diverse set of experts with different opinions, specializations, preferences can provide a decision-maker with insights about the merits of different aspects of an issue. In fact, a decision-maker may be able to perfectly learn a multidimensional state by consulting experts about different dimensions. However, consulting experts that have information about different dimensions of a decision reduces the scope for cross-verification since cross-verification is most effective when experts' information is highly correlated. Moreover, as we demonstrated in this paper, the experts have an incentive to facilitate cross-verification by acquiring correlated information. This points to an interesting tension that can inform future research on committee design. 

\bibliographystyle{ecta}
\bibliography{references}

\end{document}